%% file: main.tex
\documentclass[sigplan,screen,dvipsnames]{acmart}

\setcopyright{rightsretained}
\acmPrice{}
\acmDOI{10.1145/3437992.3439925}
\acmYear{2021}
\copyrightyear{2021}
\acmSubmissionID{poplws21cppmain-p27-p}
\acmISBN{978-1-4503-8299-1/21/01}
\acmConference[CPP '21]{Proceedings of the 10th ACM SIGPLAN International Conference on Certified Programs and Proofs}{January 18--19, 2021}{Virtual, Denmark}
\acmBooktitle{Proceedings of the 10th ACM SIGPLAN International Conference on Certified Programs and Proofs (CPP '21), January 18--19, 2021, Virtual, Denmark}

\usepackage[utf8]{inputenc}
\usepackage{microtype}
\usepackage{xspace} 
\usepackage{amsmath} 

\usepackage{amsthm}
\newtheorem{thm}{Theorem}
\newtheorem{proof}{Proof}
\newtheoremstyle{test}
  {\topsep}   
  {\topsep}   
  {\normalfont}  
  {0pt}       
  {\bfseries} 
  {.}         
  {\newline} 
  {}          
\theoremstyle{test}
\newtheorem{example}{Example}

\usepackage{setspace}
\usepackage{epsfig} 
\usepackage{tikz}
\usetikzlibrary{positioning}
\usepackage{booktabs} 
\usepackage{hyperref}
\usepackage[scaled=0.8]{beramono}  
\usepackage[T1]{fontenc}  
\usepackage{listings}
\usepackage{subcaption}
\usepackage{xcolor} 
\definecolor{light-gray}{gray}{0.95} 
\definecolor{darkgreen}{rgb}{0.0, 0.26, 0.15}

\lstdefinelanguage{HOL4}{
  morekeywords=[1]{Define, Theorem, Proof, QED, Definition, End},
  morekeywords=[2]{Type, Set, true, false, option},
  morekeywords=[3]{def, nltac, nlexplain},
  morecomment=[s]{(*}{*)},
  sensitive=true,
  basicstyle=\small\ttfamily,
  commentstyle=\color{gray},
  identifierstyle={\ttfamily\color{black}},
  keywordstyle=[1]{\ttfamily\color{violet}\bfseries},
  keywordstyle=[2]{\ttfamily\color{blue}},
  keywordstyle=[3]{\ttfamily\color{OliveGreen}\bfseries},
  escapeinside={|}{|}
}
\newcommand{\dollar}{\mbox{\textdollar}}
\newcommand*{\HOLi}[1]{\lstinline[language=HOL4,breaklines=true]{#1}}
\newcommand*{\SEMPRE}[1]{\lstinline[language=HOL4]{$#1}}

\newcommand*{\Tool}{Lassie}

\begin{document}

\title{Lassie: HOL4 Tactics by Example}
\author{Heiko Becker}
\affiliation{%
  \institution{MPI-SWS,\\ Saarland Informatics Campus (SIC)}
  \country{Germany}
}
\email{hbecker@mpi-sws.org}
\author{Nathaniel Bos}
\authornote{Nathaniel Bos was supported by a DAAD RISE Internship.}
\affiliation{%
  \institution{McGill University}
  \country{Canada}
}
\email{nathaniel.bos@mail.mcgill.ca}
\author{Ivan Gavran}
\affiliation{%
  \institution{MPI-SWS}
  \country{Germany}
}
\email{gavran@mpi-sws.org}
\author{Eva Darulova}
\affiliation{%
  \institution{MPI-SWS}
  \country{Germany}
}
\email{eva@mpi-sws.org}
\author{Rupak Majumdar}
\affiliation{%
  \institution{MPI-SWS}
  \country{Germany}
}
\email{rupak@mpi-sws.org}

\begin{abstract}
\input{sections/abstract}
\end{abstract}

\ccsdesc[500]{Software and its engineering~Formal software verification}
\ccsdesc[500]{Software and its engineering~Programming by example}
\ccsdesc[500]{Software and its engineering~Macro languages}

\keywords{Interactive Theorem Proving, HOL4, Semantic Parsing, Tactic Programming}

\maketitle

\input{sections/intro}
\input{sections/overview}
\input{sections/using_lassie}
\input{sections/tech_details}
\input{sections/euclid_proof}
\input{sections/evaluation}
\input{sections/example_sum_nlexplain}
\input{sections/hol4_tutorial}
\input{sections/related-work}
\input{sections/conclusion}

\begin{acks}
  The authors would like to thank Magnus Myreen, Zachary Tatlock, and the anonymous reviewers of ITP 2020 and CPP 2021 for providing
  feedback on \Tool{} and (initial) drafts of the paper.
  Gavran and Majumdar were supported in part by
  the DFG project 389792660 TRR 248--CPEC
  and by the European Research Council under the
  Grant Agreement 610150 (ERC Synergy Grant ImPACT).
\end{acks}

\bibliographystyle{ACM-Reference-Format}
\bibliography{biblio_clean}
\end{document}

%% file: sections/abstract.tex
Proof engineering efforts using interactive theorem proving have
yielded several impressive projects in software systems and mathematics.
A key obstacle to such efforts is the requirement that the domain expert
is also an expert in the low-level details in constructing the proof in a
theorem prover. In particular, the user needs to select a sequence
of tactics that lead to a successful proof, a task that in general requires
knowledge of the exact names and use of a large set of tactics.

We present Lassie, a tactic framework for the HOL4 theorem prover that allows
individual users to define their own tactic language \emph{by example} and 
give frequently used tactics or tactic combinations easier-to-remember names.
The core of Lassie is an extensible semantic parser, which allows the user to interactively
extend the tactic language through a process of definitional generalization.
Defining tactics in Lassie thus does not require any knowledge in implementing custom tactics,
while proofs written in Lassie retain the correctness guarantees provided by the HOL4 system.
We show through case studies how Lassie can be used in small and larger proofs by
novice and more experienced interactive theorem prover users, and how we envision it to ease the learning curve
in a HOL4 tutorial.

\keywords{Interactive Theorem Proving \and HOL4 \and Semantic Parsing.}


%% file: sections/intro.tex
\section{Introduction}\label{sec:intro}

Interactive theorem proving is increasingly replacing ``pen-and-paper''
correctness proofs in domains such as compilers~\cite{Compcert,CakeML},
operating system kernels~\cite{seL4}, and
formalized mathematics~\cite{Flyspeck,gonthier2008formal}.
Interactive theorem provers (ITPs) provide strong guarantees: all proof steps are
formalized and machine-checked by a kernel using only a small set of generally
accepted proof rules.

These guarantees come at a cost. Writing proofs in an ITP requires
both domain expertise in the target research area as well as in the particulars
of the interactive theorem prover.
Formally proving a theorem requires an expert to manually translate the general
high-level proof idea from a pen-and-paper proof into detailed, low-level kernel
proof steps, which makes writing formal proofs tedious and time-consuming.
Theorem provers thus provide tactic languages that allow to programmatically
combine low-level proof steps~\cite{ssreflect,Ltac,Eisbach,wenzel2006isabelle}.
While this makes proofs less tedious, users need to build up a vocabulary of
appropriate tactics, which constitutes a steep learning curve for novice ITP
users.

Controlled natural language interfaces~\cite{Mizar2015,NaprocheSAD} have been
explored as an alternative, more intuitive interface to an ITP. However, these
systems do not allow a combination with a general tactic language and are thus
constrained to a specific subset of proofs.

In this paper, we present the tactic framework \emph{\Tool{}} that allows HOL4
users to define their own tactic language on top of the existing ones
\emph{by example}, effectively providing an individualized interface.
Each example consists of the to-be-defined tactic (a natural language expression, called \emph{utterance})
and its definition using
existing HOL4 tactics with concrete arguments.

For instance, we can define
\begin{lstlisting}[language=HOL4, mathescape=true]
  instantiate 'x' with '$\top$'
\end{lstlisting}
as
\begin{lstlisting}[language=HOL4, mathescape=true]
  qpat_x_assum 'x' (qspec_then '$\top$' assume_tac)
\end{lstlisting}
Newly defined \Tool{} tactics map directly and transparently to the underlying
HOL4 tactics, and can be freely combined.

The main novelty to existing tactic languages is that Lassie allows to define
tactics by example and thus does not require knowledge in tactic programming.
A tactic defined by example is automatically \emph{generalized} into a parametric
tactic by \Tool{} to make the tactic applicable in different contexts, making
\Tool{} go beyond a simple macro system.

Our key technical contribution is that Lassie realizes this definition-by-example
using an extensible semantic parser~\cite{berant2013freebase,DBLP:conf/acl/WangGLM17}.
\Tool{} tactics are defined as grammar rules that map to HOL4 tactics.
\Tool{} starts with an initial core grammar that is gradually extended through user-provided examples.
For each example, the semantic parser finds matchings between the utterance and its definition.
These matchings are used to create new rules for the grammar.
Effectively, the semantic parser identifies the parameters of the newly given
command, and thus generalizes from the given example.
In our illustrative example, \Tool{} will identify \lstinline{'x'} and $\top$ as
arguments and add a rule that will work with arbitrary terms in place of
\lstinline{'x'} and $\top$.

Typically, extending a grammar through examples leads to ambiguity---for a single
uterance-definition pair there may be different possible matchings and thus
several new parsing rules introduced. In previous work~\cite{DBLP:conf/acl/WangGLM17}, this ambiguity was
resolved through user interaction, e.g. showing the user a visualization of
different parses and letting them choose the parse with the intended effect. 
However, it is non-trivial to visualize intermediate steps in a general-purpose
programming language.
Our core insight is that ITPs offer an ideal setting to resolve this ambiguity.
We show that by carefully designing the core grammar and by making use of type
information, the ambiguity can be resolved automatically.
Furthermore, ITPs ``visualize'' individual steps by showing the intermediate
proof state, and rule out wrong tactic definitions by forcing proofs to be
checked by the ITP systems kernel.

\Tool{}'s target audience are trained ITP users who implement decision procedures and simple tactic descriptions in \Tool{}.
\Tool{} allows them to define their own individualized language by defining
easy-to-remember names for individual tactics, or (frequently used) combinations
of tactics.
A tactic language implemented in \Tool{} can then used by non-expert users with
prior programming experience but without necessarily in-depth experience with an
ITP.

Compared to general tactic languages like ssreflect~\cite{ssreflect}, Ltac~\cite{Ltac}, and Eisbach~\cite{Eisbach}, \Tool{}
requires less expert knowledge, at the expense of expressiveness.
Similar to \Tool{}, structured tactic languages like Isar~\cite{Isar1999} have an extended parser.
Extending a language like Isar requires editing the source code, while \Tool{}
supports different tactic languages that can be defined simply by example.
%
While \Tool{} can be used to define a tactic language that is closer to a
natural language, by not requiring the
interface to be entirely natural, \Tool{} is more general and flexible than
systems like Mizar~\cite{Mizar2015} and Naproche-SAD~\cite{NaprocheSAD}.

We implement \Tool{} as a library for the HOL4~\cite{HOL4Tutorial} ITP system, but our
technique is applicable to other theorem provers as well. \Tool{} is fully
compatible with standard HOL4 proofs.
Since all \Tool{} tactics map to standard HOL4 tactics, \Tool{} allows 
exporting a \Tool{} proof into standard HOL4 to maintain portability of proofs.
On the other hand, the learned grammar can be ported as well and can be used, for example,
by a teacher to predefine a domain-specific (tactic) language with \Tool{}, which is
used by learners to ease proofs in a particular area.

We demonstrate \Tool{} on a number of case studies proving theorems involving logic, and
natural and real numbers.
In particular, we show the generality of the naturalized tactics by
reusing them across different proofs, and we show that \Tool{} can be
incrementally used for proofs inside larger code bases.
Finally, by predefining a tactic language with \Tool{}, we develop a tutorial
for the HOL4 theorem prover.

\paragraph{Contributions}
In summary, this paper presents:
\begin{itemize}
  \item an interactive, extensible framework called \Tool{} for writing tactics in
    an ITP by example;
  \item an implementation of this approach inside HOL4 (available at \url{https://github.com/HeikoBecker/Lassie});
  \item a number of case studies and a HOL4 tutorial (available at \url{https://github.com/HeikoBecker/HOL4-Tutorial})\\ showing the effectiveness of \Tool{}.
\end{itemize}

%% file: sections/overview.tex
\section{Lassie by Example}\label{sec:overview}


We start by demonstrating \Tool{} on a small example, before explaining our
approach in detail in~\autoref{sec:integration}.

\input{sections/example_figure}
\begin{figure}[t]
\begin{thm}
$\forall x\,y, 0 < x \wedge 0 < y \Rightarrow x^{-1} \leq y^{-1} \Leftrightarrow y \leq x$
\end{thm}
\begin{proof}
We show both sides of the implication separately.\\
To show ($x^{-1} \leq y^{-1} \Rightarrow y \leq x$), we do a case split on whether $x^{-1} < y^{-1}$ or $x^{-1} = y^{-1}$.
If $x^{-1} < y^{-1}$, the claim follows because the inverse function is inverse monotonic for $<$.
If $x^{-1} = y^{-1}$, the claim follows from injectivity of the inverse.\\
To show the case ($y \leq x \Rightarrow x^{-1} \leq y^{-1}$), we do a case split on whether $y < x$ or $y = x$.
If $y < x$ the claim follows because the inverse function is inverse monotonic for $<$. If $y = x$, the claim follows trivially.
\end{proof}
\caption{Textbook proof that the inverse function is inverse monotonic for $\leq$}
\label{subfig:real_inv_le_antimono_text}
\Description{}
\end{figure}

For our initial example we choose to prove that the inverse function ($x^{-1}$)
on real numbers
is inverse monotonic for $\leq$.
\autoref{subfig:real_inv_le_antimono_text} shows the formal statement of this
theorem, together with an (informal) proof that one may find in a textbook
(the proof uses a previously proven theorem about $<$).

\paragraph{Proofs in HOL4}
\autoref{subfig:real_inv_le_antimono_hol4} shows the corresponding HOL4 theorem
statement and proof. We can be sure that this proof is correct, because it is machine-checked by
HOL4.
%
HOL4~\cite{HOL4Tutorial} is an ITP system from the HOL-family. It is based on
higher-order logic and all proofs are justified by inference rules from a small,
trusted kernel. Its implementation language is Standard ML (SML), and similar to
other HOL provers like HOL-Light~\cite{hollight}, and
Isabelle/HOL~\cite{isabellehol}, proof steps are described using so-called
tactics that manipulate a goal state until the goal has been derived from true.

When doing a HOL4 proof, one first states the theorem to be proven and starts an
interactive proof.
\autoref{fig:example_hol} shows the example proof statement from~\autoref{subfig:real_inv_le_antimono_hol4}
on the left and the interactive session on the right. To show that
the theorem holds, the user would write a tactic proof at the
place marked with \HOLi{(* Proof *)}, starting with the initial tactic
\HOLi{rpt strip_tac}, sending each tactic to the interactive
session on the right.

A HOL4 tactic implements e.g. a single kernel step, such
as \HOLi{assume_tac thm} which introduces \HOLi{thm} as a new assumption, but a
tactic can also implement more elaborate steps, like
\HOLi{fs}, which implements a stateful simplification algorithm, and
\HOLi{imp_res_tac thm}, resolving \HOLi{thm} with the current assumptions to
derive new facts.
In our example, \HOLi{rpt strip_tac} repeatedly introduces universally quantified
variables and introduces left-hand sides of implications as assumptions.

After each tactic application, the HOL4 session prints the goal state that the
user still needs to show, keeping track of the state of the proof. Once the HOL4
session prints \HOLi{Initial goal proved}, the proof is finished. To make sure
that the proof can be checked by HOL4 when run non-interactively, the separate
tactics used in each step are chained together using the infix-operator
\HOLi{\\\\}. As this operator returns a tactic after taking some additional
inputs, it is called a tactical.

\input{sections/example_source_code}

\input{sections/example_goaltree_nlexplain}
\paragraph{Proofs in \Tool{}}
\autoref{subfig:real_inv_le_antimono_lassie} shows the proof of our theorem
using \Tool{}. This proof follows the same steps as the standard HOL4 proof, but each
tactic is called using a name that we have previously defined in \Tool{} by example.
We chose the \Tool{} tactics to be more descriptive (for us at least), and while
they make the proof slightly more verbose, they also make it easier to follow
for (non-)experts.
Each of our \Tool{} tactics maps to corresponding formal HOL4 tactics, so that the
proof is machine-checked by HOL4 as before, retaining all correctness guarantees.


Unlike existing tactic languages, \Tool{} allows to define custom tactics
\emph{by example} and thus does not require any knowledge in tactic programming.
For instance, for our example proof, we defined a new tactic by
\begin{lstlisting}
def `simplify with [REAL_LE_LT]` `fs [REAL_LE_LT]`;
\end{lstlisting}
\Tool{} automatically generalizes from this example so that we can later use this
tactic with a different argument:
\begin{lstlisting}
simplify with [REAL_INV_INJ]
\end{lstlisting}

To achieve this automated generalization, \Tool{} internally uses an extensible
semantic parser~\cite{berant2013freebase}. That is, \Tool{} tactics are
defined as grammar rules.
\Tool{} initially comes with a relatively small core grammar, supporting commonly
used HOL4 tactics. This grammar is gradually and interactively extended with
additional tactic descriptions by giving example mappings. For instance our
definition above would add the following rule to the grammar:
\begin{lstlisting}
simplify with [THM1, THM2, ...] |$\to$| fs [THM1, THM2, ...]
\end{lstlisting}
Note that this rule allows \lstinline{simplify with} to be called with a list
of theorems, not just a single theorem as in the example given.
This generalization happens completely automatically in the semantic parser and
does not require any programming by the user.

The \Tool{}-defined tactics can be used in a proof using the function
\HOLi{nltac}, that sends tactic descriptions to the semantic parser, which
returns the corresponding HOL4 tactic. Because \HOLi{nltac} has the same return
type as all other standard HOL4 tactics, it can be used as a drop-in replacement
for standard HOL4 tactics, and can be freely combined with other HOL4 tactics in
a proof.

\paragraph{Explaining Proofs with \Tool{}}
\Tool{} also comes with a function \HOLi{nlexplain}. Instead of being a drop-in
replacement, like \HOLi{nltac}, \HOLi{nlexplain} decorates the proof state
with the HOL4 tactic that is internally used to perform the current proof step.
\autoref{fig:example_goalTree} shows an intermediate state when using
\HOLi{nlexplain} to prove our example theorem.
All \Tool{} tactics inside the red dashed box on the left-hand side have been
passed to \HOLi{nlexplain}.
The goal state on the right-hand side shows the current state of the
proof as well as the HOL4 tactic script that has the same effect as the \Tool{} tactics.

We envision \HOLi{nlexplain} to be used for example in a HOL4 tutorial to ease the learning curve when
learning interactive theorem proving. \Tool{} allows a teacher to first define a
custom tactic language that follows the same structure as the HOL4 proof, but
that uses descriptive names and may be thus easier to follow for a novice. In a second step,
one can use \HOLi{nlexplain} to teach the actual underlying HOL4 tactics.

Function \HOLi{nlexplain} can furthermore be used for sharing \Tool{} proofs
without introducing additional dependencies on the semantic parser. While
sharing \Tool{} proof scripts directly is possible, it requires sharing the
state of the semantic parser as well. Alternatively, one can send the \Tool{}
proof to \HOLi{nlexplain} and obtain a HOL4 tactic script that can then be
shared without depending on the semantic parser.

\begin{figure*}[t]
\input{sections/example_grammar}
\caption{Excerpt from Lassie grammar (left) and the database (right), parsing tactics and thm list tactics}
\label{fig:ex_grammar}
\Description{}
\end{figure*}

\paragraph{More Complex Tactics}
While the target user that we had in mind when developing \Tool{} is not an
ITP expert, experts may nonetheless find \Tool{} useful to, e.g., group
commonly used combinations of tactics.
For example, to make the proofs of simple subgoals easier, an expert can define a tactic that
uses different simplification algorithms and an automated decision procedure to
attempt to solve a goal automatically:
  \begin{lstlisting}[language=HOL4]
  def `prove with [ADD_ASSOC]`
    `all_tac THEN ( fs [ ADD_ASSOC ] THEN NO_TAC)
      ORELSE (rw [ ADD_ASSOC ] THEN NO_TAC)
      ORELSE metis_tac [ ADD_ASSOC ]`
  \end{lstlisting}

The HOL4 tactic will first attempt to solve the goal using the simplification
algorithms implemented in tactics \HOLi{fs} and \HOLi{rw}, and if both fail,
it will call into the automated decision procedure \HOLi{metis_tac}, based on
first-order resolution. (Tactical \HOLi{t1 ORELSE t2} applies first tactic \HOLi{t1}, and if \HOLi{t1}
fails, \HOLi{t2} is applied. \HOLi{THEN NO_TAC} makes the
simplification fail if it does not solve the goal.)

The resulting tactic description \HOLi{prove with [THM1, THM2, ...]} is
parametric in the used list of theorems making it applicable in different contexts.

Defined tactic descriptions
are added to the grammar and are as such part of the generalization algorithm.
Thus we can reuse the just defined tactic description to define an even more elaborate version:
  \begin{lstlisting}[language=HOL4]
  def `'T' from [ CONJ_COMM ] `
    `'T' by ( prove with [CONJ_COMM] )`;
  \end{lstlisting}
This tactic description, once generalized by the semantic parser, completely
hides the fact that we may need to call into three different algorithms to prove
a subgoal, while allowing us to enrich our assumptions with arbitrary goals, as
long as they are provable by the underlying HOL4 tactics.


%% file: sections/example_figure.tex
\begin{figure*}
\begin{subfigure}{.49\textwidth}
\begin{lstlisting}[language=HOL4,mathescape=true]

Theorem REAL_INV_LE_AMONO:
  |$\forall$| x y.
    0 < x |$\wedge$| 0 < y |$\Rightarrow$|
    |$\texttt{x}^{-1}$| |$\leq$| |$\texttt{y}^{-1}$| |$\Leftrightarrow$| y |$\leq$| x
Proof
  rpt strip_tac
  \\ `|$\texttt{x}^{-1}$| < |$\texttt{y}^{-1}$| |$\Leftrightarrow$| y < x`
    by (MATCH_MP_TAC REAL_INV_LT_ANTIMONO \\ fs [])
  \\ EQ_TAC
  \\ fs [REAL_LE_LT]
  \\ STRIP_TAC
  \\ fs [REAL_INV_INJ]
QED
$\mbox{}$
\end{lstlisting}
\caption{HOL4 proof}
\label{subfig:real_inv_le_antimono_hol4}
\end{subfigure}
\begin{subfigure}{.49\textwidth}
\begin{lstlisting}[language=HOL4, mathescape=true]
Theorem REAL_INV_LE_AMONO:
  |$\forall$| x y.
    0 < x |$\wedge$| 0 < y |$\Rightarrow$|
    |$\texttt{x}^{-1}$| |$\leq$| |$\texttt{y}^{-1}$| |$\Leftrightarrow$| y |$\leq$| x
Proof
  nltac `
    introduce assumptions.
    show 'inv x < inv y <=> y < x'
      using (use REAL_INV_LT_ANTIMONO
             THEN follows trivially).
    case split.
    simplify with [REAL_LE_LT].
    introduce assumptions.
    simplify with [REAL_INV_INJ]. trivial.`
QED
\end{lstlisting}
\caption{Lassie proof}
\label{subfig:real_inv_le_antimono_lassie}
\end{subfigure}
\caption{HOL4 proof (left) and Lassie proof (right) for theorem \HOLi{REAL_INV_LE_AMONO}}
\Description{}
\end{figure*}
%

%% file: sections/example_source_code.tex
\begin{figure}[t]
\centering
  \begin{subfigure}{.25\textwidth}
    \begin{lstlisting}[language=HOL4, mathescape=true, frame=single,backgroundcolor = \color{light-gray},rulecolor=\color{gray}]
Theorem REAL_INV_LE_AMONO:
  $\forall$ x y.
    0 < x $\wedge$ 0 < y $\Rightarrow$
    (inv x $\leq$ inv y $\Leftrightarrow$ y $\leq$ x)
Proof
  rpt strip_tac

  (* Proof *)

QED
\end{lstlisting}
  \end{subfigure}
  \hspace{0.0027\textwidth}
  \begin{subfigure}{.20\textwidth}
    \begin{lstlisting}[language=HOL4, mathescape=true,frame=single,backgroundcolor = \color{light-gray},rulecolor=\color{gray}]
1 subgoal:
val it =

  0.  0 < x
  1.  0 < y
  ---------------------
  $\text{x}^{-1}$ $\leq$ $\text{y}^{-1}$ $\Leftrightarrow$ y $\leq$ x

  : proof
>
    \end{lstlisting}
  \end{subfigure}
\caption{HOL4 theorem (left) and interactive proof session (right)}\label{fig:example_hol}
\Description{}
\end{figure}


%% file: sections/example_goaltree_nlexplain.tex
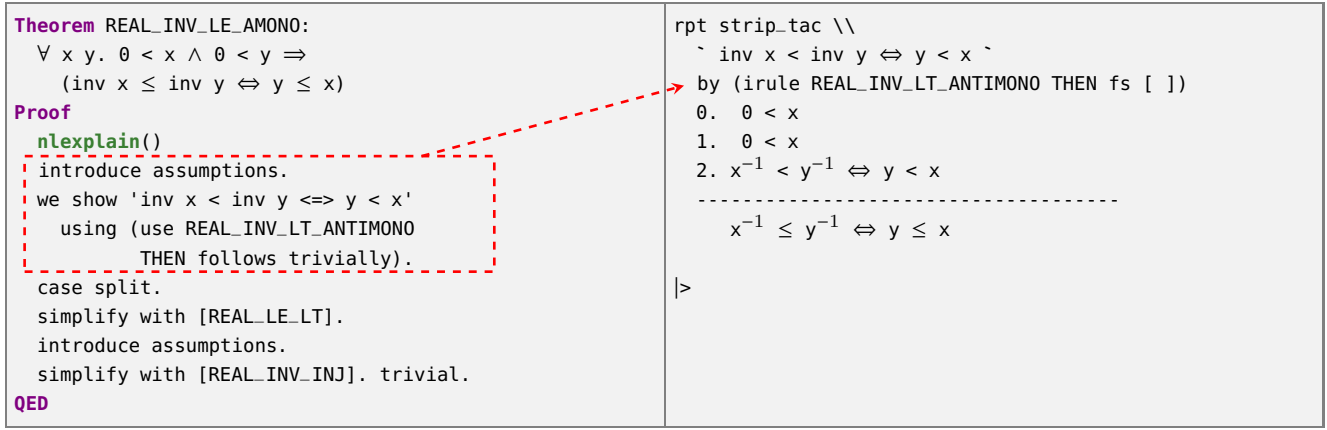
\begin{figure*}[t]
\vspace{2mm}
\centering
  \begin{subfigure}{.48\textwidth}
    \begin{lstlisting}[language=HOL4, mathescape=true, frame=single,
    escapechar=!,backgroundcolor = \color{light-gray},rulecolor=\color{gray}]
Theorem REAL_INV_LE_AMONO:
  $\forall$ x y. 0 < x $\wedge$ 0 < y $\Rightarrow$
    (inv x $\leq$ inv y $\Leftrightarrow$ y $\leq$ x)
Proof
  nlexplain()
 !\tikz[remember picture] \node [] (a) {};!introduce assumptions.           !\tikz[remember picture] \node [] (b) {};!
  we show 'inv x < inv y <=> y < x'
    using (use REAL_INV_LT_ANTIMONO
           THEN follows trivially).
  case split.
  simplify with [REAL_LE_LT].
  introduce assumptions.
  simplify with [REAL_INV_INJ]. trivial.
QED
\end{lstlisting}
  \end{subfigure}
  \hspace{0.0027\textwidth}
  \begin{subfigure}{.48\textwidth}
    \begin{lstlisting}[language=HOL4, mathescape=true,frame=single, escapechar = !,
    backgroundcolor = \color{light-gray},rulecolor=\color{gray}]
rpt strip_tac \\
  ` inv x < inv y $\Leftrightarrow$ y < x `
 !\tikz[remember picture] \node [] (c) {};!by (irule REAL_INV_LT_ANTIMONO THEN fs [ ])
  0.  0 < x
  1.  0 < x
  2. $\text{x}^{-1}$ < $\text{y}^{-1}$ $\Leftrightarrow$ y < x
  -------------------------------------
     $\text{x}^{-1}$ $\leq$ $\text{y}^{-1}$ $\Leftrightarrow$ y $\leq$ x

$|$>

$\mbox{}$
    \end{lstlisting}
  \end{subfigure}
\begin{tikzpicture}[remember picture, overlay,
    every edge/.append style = { ->, thick, >=stealth,
                                  red, dashed, line width = 1pt }]
  \draw [dashed, red, line width = 1pt] (a) + (-0.1, 0.2) rectangle (-11.0,-0.55);
  \draw (b) + (0.0, 0.2) edge (c);
\end{tikzpicture}
\caption{Intermediate proof state using goalTree's and \HOLi{nlexplain}}
\label{fig:example_goalTree}
\Description{}
\end{figure*}

%% file: sections/example_grammar.tex
\begin{minipage}[t]{.69\textwidth}
\begin{tabbing}
\hspace*{5.5em}\=\hspace*{9em}\=\kill
\SEMPRE{ROOT} \>$\rightarrow$ \SEMPRE{tactic} \>$(\lambda x. x)$\\
\SEMPRE{tactic} \>$\rightarrow$ \SEMPRE{TOKEN} \>$(\lambda x.$\texttt{\small lookup "tactic"} $x)$\\
\SEMPRE{tactic} \>$\rightarrow$ \SEMPRE{thm->tactic} \SEMPRE{thm} \>$(\lambda x\,y. x\,y)$\\
\SEMPRE{thm->tactic} \>$\rightarrow$ \SEMPRE{TOKEN} \>$(\lambda x.$\texttt{\small lookup "thm list->tactic"} $x)$\\
\SEMPRE{thm} \>$\rightarrow$ \SEMPRE{TOKEN} \> $(\lambda x. x)$
\end{tabbing}
\end{minipage}
\hspace{1em}
\vrule
\hspace{1em}
\begin{minipage}[t]{.29\textwidth}
\begin{tabbing}
\hspace*{4em}\=\hspace*{1em}\=\kill
\HOLi{gen_tac} \>\texttt{:} \> \HOLi{tactic}\\
\HOLi{all_tac} \>\texttt{:} \> \HOLi{tactic}\\
\HOLi{strip_tac} \> \texttt{:} \> \HOLi{tactic}\\
\HOLi{fs} \> \texttt{:} \> \HOLi{thm list->tactic}\\
\HOLi{simp} \> \texttt{:} \> \HOLi{thm list->tactic}
\end{tabbing}
\end{minipage}

%% file: sections/using_lassie.tex
\section{Defining Tactics in \Tool{}}\label{sec:integration}


Existing approaches to tactic languages, like Eisbach~\cite{Eisbach} and
ssreflect~\cite{ssreflect} are implemented as domain-specific languages (DSL), usually
within the theorem prover's implementation language.
In these approaches, defining a new tactic is the same as defining a function in
the implemented DSL.
If a tactic should be generalized over e.g.\ a list of theorems, this
generalization must be performed manually by the user of the tactic language.

In contrast, \Tool{}'s tactics are defined in a grammar that is extended
interactively by example using a semantic parser~\cite{berant2013freebase} that performs
parameter generalization automatically.
We define an initial core grammar (\autoref{subsec:Grammar}) that users can
extend by example (\autoref{subsec:extendingLassie}). Each such defined description
(\Tool{} tactic), maps a description to a (sequence of) HOL4 tactics, which is
then applied to the proof state and checked by the HOL4 kernel.
Note that a \Tool{} user does not directly modify and thus does not have to be
aware of the underlying (core) grammar---the extension happens by example.



\subsection{The Core Grammar}
\label{subsec:Grammar}


The left-hand side of \autoref{fig:ex_grammar} shows a subset of \Tool{}'s core
grammar.
\SEMPRE{ROOT} is the symbol for the root node in the grammar and must always be a valid tactic.
The core grammar is used to parse theorems, tactics, tacticals (of type
\HOLi{thm list -> tactic}) and looks up functions of these types.

Each rule has the form \SEMPRE{left} $\rightarrow$ \SEMPRE {right} $(\lambda x. \ldots)$.
While \SEMPRE{left} $\rightarrow$ \SEMPRE{right} works just as in a standard context free grammar,
the $\lambda$-abstraction, called logical form, is applied to the result of
parsing \SEMPRE{right} using the grammar.
The logical form allows us to manipulate parsing results after they have been
parsed by the grammar, essentially interpreting them within the parser.
In \Tool{} we use it to implement function applications when combining tactics,
and to lookup names in a database.

We have built a core grammar for \Tool{} that supports the most common tactics
and tacticals of HOL4. For instance the core grammar will parse
\HOLi{fs [REAL_INV_INJ]} unambiguously into the equivalent SML code as its logical form.
We think of this core grammar as the starting point for users to define
\Tool{} tactics on top of the HOL4 tactics.

Adding every HOL4 tactic and tactical as a separate terminal
to the grammar would clutter it unnecessarily and make it hard to maintain.
That is why the grammar allows so-called lookup rules that check a dictionary for
elements of predefined sets.
The right-hand side of \autoref{fig:ex_grammar} shows a subset of the database
used for the lookups.
In the grammar in \autoref{fig:ex_grammar}, a tactic can then either be looked
up from the database (second rule), or a tactic can be a combination of a function
of type \HOLi{thm -> tactic} and a theorem (third rule).
We refer to functions of type \HOLi{thm -> tactic} as theorem tactics, as they
take a theorem as input, and return a HOL4 tactic.
Theorem tactics are again looked up from the database, whereas theorems can be
any possible string denoted in the grammar by \SEMPRE{TOKEN}.
In addition to HOL4 tactics and theorem tactics, our core grammar also uses a
combination of rules (not shown in \autoref{fig:ex_grammar}) to support
functions that return a tactic of type
\begin{itemize}
\item \HOLi{thm list -> tactic}
\item \HOLi{tactic -> tactic}
\item \HOLi{term quotation -> tactic}
\item \HOLi{(thm -> tactic) -> tactic}
\item \HOLi{tactic -> tactic -> tactic}
\item \HOLi{term quotation -> (thm -> tactic)}~\HOLi{-> thm -> tactic}
\item \HOLi{term quotation list -> (thm -> tactic)}~\HOLi{->}

    \HOLi{thm -> tactic}
\end{itemize}
These types capture most of the tactics implemented in HOL4, and we add a
subset of 53 commonly used tactics into the database.

\paragraph{Non-Ambiguity}
A common issue in semantic parsing is grammar ambiguity. In \Tool{}, having an
ambiguous grammar is not desirable as it would require users to disambiguate
each ambiguous \Tool{} tactic while proving theorems.
We thus aim to have an unambiguous grammar and achieve this by a careful design
of our core grammar.
By encoding the types of the tactics as non-terminals, our core grammar
acts as a type-checker for our supported subset of HOL4 tactics.
Even after defining custom tactics, the semantic parser will always parse \Tool{} tactics into
the subset it can type check thus  keeping the grammar unambiguous.
During our experiments we have not found a case where extending the grammar
introduced any ambiguity, which reassures this design choice.


\subsection{Extending Lassie with New Definitions}
\label{subsec:extendingLassie}

With our core grammar, \Tool{} can parse the HOL4 tactics we have added to the
grammar into their (equivalent) SML code. We now explain how this grammar can be
interactively extended by example in order to provide custom names for
(sequences of) tactics.

\Tool{}'s tactic learning mechanism relies on a semantic parser.
A semantic parser converts a natural language utterance into a
corresponding (executable) logical form or---due to ambiguity---a ranked list of candidates.
Semantic parsers can be implemented in many ways, e.g., they can be rule-based or learned
from data~\cite{DBLP:journals/cacm/Liang16}.
SEMPRE~\cite{berant2013freebase}, which we use, is a toolkit for developing
semantic parsers for different tasks.
It provides commonly used natural language processing methods, and
different ways of encoding logical forms.

\Tool{}'s semantic parser is implemented on top of the interactive version of
SEMPRE~\cite{DBLP:conf/acl/WangGLM17}.
It starts with a core formal grammar, which can be expanded through interactions
with the user.
Users can add new concepts to the grammar by example using \Tool{}'s library
function \HOLi{def}, which invokes the semantic parser. Each example consists
of a (\emph{utterance}, \emph{definition}) pair, where the utterance is the new
tactic to be defined and the definition is an expression that is already part
of the grammar.
For instance, we can give as example:
\begin{lstlisting}
  def `simplify with REAL_ADD_ASSOC`   (*utterance*)
    `fs [REAL_ADD_ASSOC]`            (*definition*)
\end{lstlisting}
Note that the command demonstrates the new tactic (\HOLi{simplify with}) with a particular
argument (\HOLi{REAL_ADD_ASSOC}), but does not explicitly state what the argument is.

The definition has to already be part of the grammar and thus fully parsable,
otherwise the parser will reject the pair, whereas only some parts of the
utterance may be parsable. That is, the definition needs to be already understood
by the semantic parser, either because it is part of the core grammar or because
it was previously already defined by the user.



The function \HOLi{def} first obtains a logical form for the definition (which
exists since the definition is part of the grammar).
The semantic parser then
induces one or more grammar rules from the utterance-definition pair
and attaches the logical form of the definition to those rules.





The induction of new grammar rules relies on
finding correspondences between parsable parts of the utterance and its definition.
As an example, observe our \HOLi{simplify with} command.
Because \lstinline{REAL_ADD_ASSOC} can be parsed into a category \lstinline{$thm}, the two new production
rules added to the grammar are:
\begin{lstlisting}[escapeinside={(*}{*)}]
$tactic (*$\rightarrow$*)
  simplify with REAL_ADD_ASSOC ((*$\lambda$*) x.fs [REAL_ADD_ASSOC])
$tactic (*$\rightarrow$*)
  simplify with $thm ((*$\lambda$*) thm. fs [thm])
\end{lstlisting}

Based on the second added rule, we can now use the \Tool{} tactic
\HOLi{simplify with} connected to any other description that is parsed as a
\SEMPRE{thm}, because the parser identified
\HOLi{REAL_ADD_ASSOC} as an argument and generalized from our example by learning
the $\lambda$-abstraction over the variable \HOLi{thm}.

Next time the user calls, for instance,
\begin{lstlisting}
  nltac `simplify with REAL_ADD_COMM`
\end{lstlisting}
\Tool{}'s semantic parser will parse this command into the tactic \HOLi{fs
[REAL_ADD_COMM]} using the second added rule.

%% file: sections/tech_details.tex
\section{\Tool{} Design}\label{sec:communication}


Lassie is implemented as a HOL4 library, which can be loaded into a running HOL4 session
with \HOLi{open LassieLib;}. This will start a SEMPRE
process and the library captures its input and output as SML streams.
Whenever \HOLi{nltac} or \HOLi{nlexplain} are run, the input is send to SEMPRE
over the input stream, and if it can be parsed with the currently learned
grammar, SEMPRE writes the resulting HOL4 tactic to the output stream as a
string.
If parsing fails, i.e. SEMPRE does not recognize the description, LassieLib
raises an exception, such that an end-user can define the tactic with a call to
\HOLi{def}.

We want \HOLi{nltac} to act as a drop-in replacement for HOL4 tactics.
Therefore, \HOLi{nltac} must not only be able to parse single tactics, but must
also be able to parse full tactic scripts, performing a proof from start to finish.
During our case-studies, we noticed that SEMPRE was not built for parsing
large strings of text, but rather for smaller examples.
To speed up parsing, we have defined a global constant, \HOLi{LassieSep} which
is used to split input strings of \HOLi{nltac}.
For example, calling
\begin{lstlisting}
  nltac `case split. simplify with [REAL_LE_LT].`
\end{lstlisting}
will lead to two separate calls to the semantic parser:
one for \HOLi{case split} and one for \HOLi{simplify with [REAL_LE_LT]}.
The resulting HOL4 tactics are joined together using the \HOLi{THEN_LT} tactical,
which is a more general version of the tactical \HOLi{\\\\}, as it has an
additional argument for selecting the subgoal to which the given tactic is applied.
When proving a goal interactively, some tactics, like induction, and case splitting,
can lead to multiple subgoals being generated.
We use the \HOLi{THEN_LT} tactical to implement selecting subgoals in
\HOLi{nltac}.


\begin{sloppypar}
There are some differences in how \HOLi{nltac} and \HOLi{nlexplain}
are used.
Function \HOLi{nltac} can be used as a drop-in replacement for HOL4 tactics,
and thus supports selection of subgoals.
In contrast, \HOLi{nlexplain} is meant to be used interactively, and therefore
parses \Tool{} tactics, but does not support selection of subgoals.
Instead, subgoals are proven in order of appearance.
The main purpose of \HOLi{nlexplain} is to show how \Tool{} tactics are
translated back into HOL4 tactics.
To do so, it modifies HOL4's interactive read-eval-print loop (REPL), and thus can only be used interactively,
but not to replace plain HOL4 tactics in proof scripts like \HOLi{nltac}.
\end{sloppypar}



To differentiate between SML expressions and HOL4 expressions,
HOL4 requires HOL4 expressions to be wrapped in quotes (\`{}), but
quotes are also a way of allowing multiline strings in HOL4 proofscripts.
Therefore we choose quotes to denote the start and end of a \Tool{} proofscript,
and use apostrophes (\HOLi{'}) to denote the start and the end
of a HOL4 expression in a \Tool{} proof script.

\Tool{} currently does not support debugging tactic applications.
While an end-user can easily define new tactics by example using the semantic
parser, figuring out the tactics exact behavior, and fixing bugs still requires
the user to manually step through the corresponding HOL4 tactic in an
interactive proof and manually inspecting steps.
We see extending \Tool{} with debugging support as future work.



\subsection{Extending Lassie with New Tactics}
\label{subsec:newTacs}
Our initial core grammar supports only a fixed set of the most commonly used
HOL4 tactics.
However, it is common in ITPs to develop custom tactics on a per-project basis,
possibly including fully blown decision procedures~\cite{FPTaylordecisionprocedure}.
To make sure that users can add their own HOL4 tactics as well as custom
decision procedures to \Tool{}, the library provides the functions
\HOLi{addCustomTactic}, \HOLi{addCustomThmTactic}, and \HOLi{addCustomThmlistTactic}.

The difference between \HOLi{def} and \HOLi{addCustom[*]Tactic} is in where the
elements are added to the semantic parser's grammar.
Function \HOLi{def} uses SEMPRE's generalization algorithm and adds rules
to the grammar that may contain non-terminals (e.g. \lstinline{follows from [ $thms ]}). 
Function \HOLi{addCustomTactic} always adds a new terminal to the
grammar.

We explain \HOLi{addCustomTactic} by example.
Suppose a user wants to reuse an existing
linear decision procedure for real numbers (\HOLi{REAL_ASM_ARITH_TAC}) to
close simple proof goals.
Running
  \lstinline{addCustomTactic REAL_ASM_ARITH_TAC}
adds the new
production rule \SEMPRE{tactic} $\rightarrow$ \HOLi{REAL_ASM_ARITH_TAC} to the
SEMPRE grammar.
Tactic \HOLi{REAL_ASM_ARITH_TAC} can then be used in subsequent calls to \HOLi{def} to
provide \Tool{}-based descriptions, or immediately in \HOLi{nltac} and \HOLi{nlexplain}.

Now that SEMPRE accepts the decision procedure as a valid tactic, we extend
our expert automation tactic from before to try to solve a goal with this
decision procedure too:
  \begin{lstlisting}
  def `prove with [ADD_ASSOC]`
    `all_tac THEN ( fs [ ADD_ASSOC ] THEN NO_TAC)
      ORELSE (rw [ ADD_ASSOC ] THEN NO_TAC)
      ORELSE REAL_ASM_ARITH_TAC
      ORELSE metis_tac [ ADD_ASSOC ]`
  \end{lstlisting}

Functions \HOLi{addCustomThmTactic}, and \HOLi{addCustomThmlistTactic} work
similarly, adding grammar rules for \SEMPRE{thm->tactic} and \SEMPRE{thm list->tactic}.

\subsection{Defining and Loading Libraries}
\label{subsec:libs}

Users can define libraries with their own defined \Tool{} tactics using the
function \HOLi{registerLibrary} which takes as first input a string, giving the
libraries a unique name, and as second input a function of type \HOLi{:unit ->
unit}, where the function should call \HOLi{def} on the definitions to be added,
following~\autoref{subsec:extendingLassie}.
The defined libraries can then be shared and loaded simply by calling the function
\HOLi{loadLibraries}.

We defined libraries for proofs using logic, natural numbers, and real numbers
from our case studies and used these in our HOL4 tutorial  (\autoref{sec:evaluation})




%% file: sections/euclid_proof.tex
\begin{figure*}[t]
\begin{subfigure}{.49\textwidth}
\begin{lstlisting}[language=HOL4, mathescape=true]
Theorem EUCLID:
  $\forall$ n . $\exists$ p . n < p $\wedge$ prime p
Proof
  CCONTR_TAC \\ fs[]
  \\ `FACT n + 1 $\neq$ 1`
      by rw[FACT_LESS, neq_zero]
  \\ qspec_then `FACT n + 1` assume_tac PRIME_FACTOR
  \\ `$\exists$ q. prime q $\wedge$ q divides (FACT n + 1)` by fs[]
  \\ `q $\leq$ n` by metis_tac[NOT_LESS_EQUAL]
  \\ `0 < q` by metis_tac[PRIME_POS]
  \\ `q divides FACT n`
      by metis_tac [DIVIDES_FACT]
  \\ `q = 1` by metis_tac[DIVIDES_ADDL, DIVIDES_ONE]
  \\ `prime 1` by fs[]
  \\ fs[NOT_PRIME_1]
QED
$\mbox{}$
\end{lstlisting}
\end{subfigure}
\begin{subfigure}{.49\textwidth}
\begin{lstlisting}[language=HOL4, mathescape=true]
Theorem EUCLID: (* Lassie *)
  $\forall$ n . $\exists$ p . n < p $\wedge$ prime p
Proof
  nltac`
    suppose not. simplify.
    we can derive 'FACT n + 1 <> 1'
      from [FACT_LESS, neq_zero].
    thus PRIME_FACTOR for 'FACT n + 1'.
    we further know
      '$\exists$ q. prime q and q divides (FACT n + 1)'.
    show 'q <= n' using [NOT_LESS_EQUAL].
    show '0 < q' using [PRIME_POS] .
    show 'q divides FACT n' using [DIVIDES_FACT].
    show 'q=1' using [DIVIDES_ADDL, DIVIDES_ONE].
    show 'prime 1' using (simplify).
    [NOT_PRIME_1] solves the goal.`
QED
\end{lstlisting}
\end{subfigure}
\caption{HOL4 proof (left) and \Tool{} proof (right) of euclids theorem}
\label{fig:euclid}
\Description{}
\end{figure*}

%% file: sections/evaluation.tex
\section{Case Studies}\label{sec:evaluation}
We evaluate \Tool{} on three case studies and show how it can be used for
developing a HOL4 tutorial.
In the paper, we show only the main theorems for the case studies, but the full
developments can be found in the \Tool{} repository.

\subsection{Case Study: Proving Euclid's Theorem}

First, we prove Euclid's theorem from the HOL4 tutorial~\cite{HOL4Tutorial} that is distributed with the HOL4 theorem prover documentation.
Euclid's theorem states that the prime numbers form an infinite sequence.
Its HOL equivalent states that for any natural number $n$, there exists a
natural number $p$ which is greater than $n$ and a prime number.

To prove the final theorem, shown in~\autoref{fig:euclid}, we
have proven 19 theorems in total.
To prove these theorems, we defined a total of 22 new
tactics using \HOLi{LassieLib.def}.
Some tactics have been used only once, but for example the tactic
\HOLi{[...] solves the goal}, was reused 16 times.

Another example is the tactic \HOLi{thus PRIME_FACTOR for 'FACT n + 1'} which
introduces a specialized version of the theorem \HOLi{PRIME_FACTOR}, proving the
existence of a prime factor for every natural number.
Note how the tactic description can freely mix text descriptions with the
parameters for the underlying tactic.
Similarly, the first step of the HOL4 proof reads \HOLi{CCONTR_TAC}, which
initiates a proof by contradiction.
For an untrained user, figuring out and remembering this name can be cumbersome,
even though the user might know the high-level proof step.
Instead, in \Tool{} we have used the---for us---more intuitive name
\HOLi{suppose not}.

Finally, each sub-step of the HOL4 proof is closed using the tactic \HOLi{metis_tac}.
For an expert user, it is obvious that \HOLi{metis_tac} can be used, because the
expert knows that it performs first order resolution to prove the goal.
In the \Tool{} proof, we hide \HOLi{metis_tac []} in combination with the
simplification tactics \HOLi{fs []} and \HOLi{rw[]} under the description
\HOLi{[] solves the goal}.
To further automate proving simple subgoals, we combine the tactic
\HOLi{[] solves the goal} with our \Tool{} tactic for proving subgoals (\HOLi{show 'T' using (gen_tac)})
by defining \HOLi{show 'T' using [...]} as\\ \HOLi{show 'T' using ([...] solves the goal)}.

\subsection{Case Study: Real and Natural Number Theorems}
Next, we will show how Lassie can be used in more involved
proofs about both real and natural numbers.
As an example, we prove that for any natural number $n$, the sum of the cubes of the
first $n$ natural numbers is the same as the square of the sum.
The Lassie proof of the final theorem is in \autoref{fig:sum_of_cubes_lassie}.

\begin{figure}[t]
\begin{lstlisting}[language=HOL4, mathescape=true]
Theorem sum_of_cubes_is_squared_sum:
   $\forall$ n. sum_of_cubes n = (sum n) pow 2
Proof
  nltac `
    induction on 'n'.
    simplify conclusion with [sum_of_cubes_def, sum_def].
    rewrite with [POW_2, REAL_LDISTRIB, REAL_RDISTRIB,
      REAL_ADD_ASSOC].
    showing
      '&SUC n pow 3 =
       &SUC n * &SUC n + &SUC n * sum n + sum n * &SUC n'
      closes the proof
      because (simplify conclusion with [REAL_EQ_LADD]).
    we know '& SUC n * sum n + sum n * &SUC n =
      2 * (sum n * & SUC n)'.
    rewrite once [<- REAL_ADD_ASSOC].
    rewrite last assumption.
    rewrite with [pow_3, closed_form_sum, real_div,
      REAL_MUL_ASSOC].
    we know '2 * &n * (1 + &n) * inv 2 =
      2 * inv 2 * & n * (1 + &n)'.
    rewrite last assumption.
    simplify conclusion with [REAL_MUL_RINV].
    we show 'n + 1 = SUC n' using (simplify conclusion).
    rewrite last assumption. simplify conclusion.
    we show '2 = (SUC (SUC 0))'
      using (simplify conclusion).
    rewrite last assumption. rewrite last assumption.
    rewrite with [EXP].
    we show 'SUC n = n + 1' using (simplify conclusion).
    rewrite last assumption.
    rewrite with [GSYM REAL_OF_NUM_ADD, pow_3].
    rewrite with [REAL_OF_NUM_ADD, REAL_OF_NUM_MUL,
                  MULT_RIGHT_1, RIGHT_ADD_DISTRIB,
                  LEFT_ADD_DISTRIB, MULT_LEFT_1].
    simplify.`
QED
\end{lstlisting}
\caption{Lassie proof that the sum of the natural numbers from $1$ to $n$ cubed is the same as the square of their sum}
\label{fig:sum_of_cubes_lassie}
\Description{}
\end{figure}

We have proven a total of 5 theorems:
two (real-numbered) binomial laws, the closed form for summing the first $n$
natural numbers, a side lemma on exponentiation, and the main result about
cubing the first $n$ numbers.
All our proofs in this case study have been performed using the HOL4 theory of
real numbers simply for convenience, as we found real number arithmetic easier
for proving theorems that involve subtractions, powers, and divisions.
We defined a total of 42 tactics by example using \HOLi{LassieLib.def} and added
3 custom tactics using \HOLi{LassieLib.addCustomTactic} and
\HOLi{LassieLib.addCustomThmTactic}.
Again, some of the tactics were used only once or twice but our \Tool{}
tactics for rewriting with a theorem (two calls to \HOLi{LassieLib.def} to
support rewriting from left to right, and right to left) are reused 13 times
within the proofs.

This \Tool{} proof shows how it can be extended with custom tactics.
Our restricted core grammar of \Tool{} does not include HOL4's decision
procedure for reals.
Nevertheless, a user may want to provide this tactic as part of some
automation.
Because \Tool{} supports on-the-fly grammar extensions we add the decision
procedure for reals (\HOLi{REAL_ASM_ARITH_TAC}) to the grammar:
\HOLi{addCustomTactic REAL_ASM_ARITH_TAC}.
Having added this tactic, it can be used just like the HOL4 tactics we support
in the base grammar.
Thus we define a \Tool{} tactic using the decision procedure:
\begin{lstlisting}[language=HOL4, mathescape=true]
def `we know 'T'`
  `'T' by (REAL_ASM_ARITH_TAC ORELSE DECIDE_TAC)`
\end{lstlisting}
The semantic parser now automatically generalizes the grammar rule for this tactic, learning the rule
\begin{lstlisting}[language=HOL4, mathescape=true]
$\dollar$tactic |$\to$|
  we know '$\dollar$term'($\lambda$ t.
    't' by (REAL_ASM_ARITH_TAC ORELSE DECIDE_TAC))
\end{lstlisting}
With this, we can use more complicated tactics like
\HOLi{we know '2 * &n * (1 + &n) * inv 2 = 2 * inv 2 * &n * (1 = &n)'}.

In general, combining the extensibility of \Tool{} and the generalization of
SEMPRE allows us to support arbitrary settings where trained experts can
implement domain-specific decision procedures and provide simple tactic
descriptions to novice users that want to use them in a HOL4 proof, essentially
decoupling the automation from its implementation.
Equally, any user can define personalized and more intuitive names for
often-used tactics.

\subsection{Case Study: Naturalizing a Library Proof}
In our final example, we show how Lassie can be integrated into larger developments,
by proving a soundness theorem from a library of FloVer~\cite{Flover}.
FloVer is a verified checker for finite-precision roundoff error bounds implemented in HOL4.
Its HOL4 definitions and proofs span approximately 10000 lines of code
and the interval library is one of the critical components which is used
in most of the soundness proofs.
As the FloVer proofs are performed over real numbers, we reuse the tactic
descriptions from our previous example and do not need to add additional
definitions.
In \autoref{fig:flover_iv_inv} we show that if we have an interval $iv$, and a
real number $a \in iv$, then the inverse of $a$ is contained in the inverse of
$iv$.

\begin{figure}[t]
\begin{lstlisting}[language=HOL4, mathescape=true]
Theorem interval_inversion_valid:
  $\forall$ iv a.
    (SND iv < 0 \/ 0 < FST iv) /\ contained a iv ==>
    contained (inv a) (invertInterval iv)
Proof
  nltac `
  introduce variables.
  case split for 'iv'.
  simplify with [contained_def, invertInterval_def].
  introduce assumptions.
  rewrite once [<- REAL_INV_1OVER].
  Next Goal.
    rewrite once [ <- REAL_LE_NEG].
    we know 'a < 0'. thus 'a <> 0'.
    we know 'r < 0'. thus 'r <> 0'.
    'inv(-a) <= inv (-r) <=> (- r) <= -a' using
      (use REAL_INV_LE_AMONO THEN simplify).
    resolve with REAL_NEG_INV.
    rewrite assumptions.
    follows trivially.
  Next Goal.
    rewrite once [<- REAL_LE_NEG].
    we know 'a < 0'. thus 'a <> 0'. we know 'q <> 0'.
    resolve with REAL_NEG_INV.
    'inv (-q) <= inv (-a) <=> (-a) <= (-q)' using
      (use REAL_INV_LE_AMONO THEN simplify
       THEN trivial).
    rewrite assumptions. follows trivially.
  Next Goal.
    rewrite with [<- REAL_INV_1OVER].
    'inv r <= inv a <=> a <= r' using
      (use REAL_INV_LE_AMONO THEN trivial).
    follows trivially.
  Next Goal.
    rewrite with [<- REAL_INV_1OVER].
    'inv a <= inv q <=> q <= a' using
      (use REAL_INV_LE_AMONO THEN trivial).
    follows trivially.`
QED
\end{lstlisting}
\caption{Soundness of FloVer's interval inversion in Lassie}
\label{fig:flover_iv_inv}
\Description{}
\end{figure}

This example shows that Lassie's tactic definitions are expressive enough to
build libraries of common tactic descriptions that can be shared between
projects. 

%% file: sections/example_sum_nlexplain.tex
\begin{figure*}[t]
  \vspace{2mm}
\centering
  \begin{subfigure}{.48\textwidth}
    \begin{lstlisting}[language=HOL4, mathescape=true, frame=single,escapechar=!,backgroundcolor = \color{light-gray},rulecolor=\color{gray}]
Definition sum_def:
  sum (n:num) = if n = 0 then 0 else sum (n-1) + n
End

Theorem closed_form_sum:
  $\forall$ n. sumEq n = n * (n + 1) DIV 2
Proof
  nlexplain()
 !\tikz[remember picture] \node [] (a2) {};!Induction on 'n'.         !\tikz[remember picture] \node [] (b2) {};!
  simplify with [sumEq_def].`
  simplify with [sumEq_def, GSYM ADD_DIV_ADD_DIV].
  '2 * SUC n + n * (n + 1) = SUC n * (SUC n + 1)'
    suffices to show the goal.
  show 'SUC n * (SUC n + 1) =
        (SUC n + 1) + n * (SUC n + 1)'
    using (simplify with [MULT_CLAUSES]).
  simplify.
  show 'n * (n + 1) = SUC n * n'
    using (trivial using [MULT_CLAUSES,MULT_SYM]).
  rewrite assumptions. simplify.
QED
\end{lstlisting}
  \end{subfigure}
  \hspace{0.0027\textwidth}
  \begin{subfigure}{.48\textwidth}
    \begin{lstlisting}[language=HOL4, mathescape=true,frame=single,escapechar=!,backgroundcolor = \color{light-gray},rulecolor=\color{gray}]
Induct on ` n `
  >- ( fs [ sum_def ])
  >- ( fs [ sum_def, GSYM ADD_DIV_ADD_DIV ] \\
 !\tikz[remember picture] \node [] (c2) {};!    `2 * SUC n + n * (n + 1) = SUC n * (SUC n + 1)`
        suffices_by (fs [ ]) \\
  0. sum n = n * (n + 1 DIV 2)
  ----------------------------
  2 * SUC n + n * (n + 1) = SUC n * (SUC n + 1)

$|$>

$\mbox{}$
    \end{lstlisting}
  \end{subfigure}
\begin{tikzpicture}[remember picture, overlay,
    every edge/.append style = { ->, thick, >=stealth,
                                  red, dashed, line width = 1pt }]
  \draw [dashed, red, line width = 1pt] (a2) + (-0.1, 0.2) rectangle (-10.0,-0.75);
  \draw (b2) + (0.0, 0.2) edge (c2);
\end{tikzpicture}
\caption{Intermediate state of \HOLi{nlexplain} in our tutorial}\label{fig:nlexplain_tutorial}
\Description{}
\end{figure*}


%% file: sections/hol4_tutorial.tex
\subsection{HOL4 Tutorial}\label{sec:hol4_tutorial}

We have used \Tool{} to write a new tutorial for HOL4
 with the goal of
decoupling the learning of the basic structure of formal proofs from the
particular syntax and tactic names of HOL4, and by this easing the learning
curve.
Our tutorial is based on the existing HOL4 tutorial~\cite{HOL4Tutorial} and the
HOL4 emacs interaction guide.

First, the new HOL4 user uses \HOLi{nltac} and the \Tool{} tactics
that we defined for our three case studies (i.e.\ loads them as libraries) to do
the proofs. He or she can thus learn the syntax of theorems and definitions,
as well as structure of proofs without having to also learn the often
unintuitive tactic names of the proofs.
For example, we show the proof of the closed form for summing the first $n$
natural numbers from our tutorial in \autoref{fig:gauss}.
The example proof shows \Tool{} tactics that abstract from the tactic, but not the theorem names.
\Tool{} has limited support for defining descriptions of theorems similar to how \Tool{} tactics are defined
which could be used when developing individual languages.

In the second step, the new HOL4 user is introduced to the HOL4 tactics using
\HOLi{nlexplain}. For instance, they can step through the proof and see the
HOL4 tactics underlying each \Tool{} tactic.
We show an example in \autoref{fig:nlexplain_tutorial}.
The left-hand side shows the HOL4 proof state obtained by applying \Tool{}
tactics with \HOLi{nlexplain},
and the right-hand side the modified HOL4 REPL with the current proof goal and a
partial HOL4 tactic script.
The red dashed box on the left-hand side marks all \Tool{} tactics that have been
passed to \HOLi{nlexplain}.

Our tutorial is split into six separate parts.
We start by explaining how HOL4 (and \Tool{}) are installed and configured on a
computer such that the tutorial can be followed interactively.
Next, we explain how one interacts with HOL4 in an interactive session.
The first technical section uses the proof from \autoref{fig:gauss} as a first
example of an interactive HOL4 proof, using only \HOLi{nltac} to perform proofs.
Having introduced the reader to the basics of interactive proofs in HOL4, we
show how a simple library of proofs can be developed.
The library is a re-implementation of our first case study, and hence follows
the structure of the original HOL4 tutorial.
It spans a total of two definitions, and 13 theorems.
For each of the theorems we show a proof using \HOLi{nltac}.
Only after these introductory sections, where a user will have already gained an
intuition both about how one interacts with the HOL4 REPL, and how proofs are stored
in reusable theories, the next section introduces \HOLi{nlexplain} and explains
how HOL4 proofs are performed with plain HOL4 tactics.
Finally, the tutorial concludes with some helpful tips and tricks that we have
collected.

We defined the tutorial using definitions that we personally found intuitive.
However, \Tool{}'s ability to define tactics by example allows each teacher to
define their own individual language in a straightforward way.
\begin{figure}[t]
\begin{lstlisting}[language=HOL4, mathescape=true]
Theorem closed_form_sum:
  $\forall$ n. sum n = (n * (n + 1)) DIV 2
Proof
  nltac`
   Induction on 'n'.
   Goal 'sum 0 = 0 * (0 + 1) DIV 2'.
     simplify.
   End.
   Goal 'sum (SUC n) = SUC n * (SUC n + 1) DIV 2'.
     use [sum_def, GSYM ADD_DIV_ADD_DIV] to simplify.
     '2 * SUC n + n * (n + 1) = SUC n * (SUC n + 1)'
       suffices to show the goal.
     show 'SUC n * (SUC n + 1) =
          (SUC n + 1) + n * (SUC n + 1)'
       using (simplify with [MULT_CLAUSES]).
     simplify.
     show 'n * (n + 1) = SUC n * n'
       using (trivial using [MULT_CLAUSES, MULT_SYM]).
     '2 * SUC n = SUC n + SUC n' follows trivially.
     'n * (SUC n + 1) = SUC n * n + n' follows trivially.
     rewrite assumptions. simplify.
   End.`
QED
\end{lstlisting}
\caption{Example proof of the closed form for summing $n$ numbers using \Tool{} in our HOL4 tutorial}\label{fig:gauss}
\Description{}
\end{figure}

%% file: sections/related-work.tex
\section{Related Work}\label{sec:related}

In this section, we review approaches designed to ease
the user burden when writing proofs in an ITP.

\paragraph{Hammers}
So-called ``hammers'' use automated theorem provers (ATP) to discharge
proof obligations by translating a proof goal into the logic of
an ATP and a proof back into the logic of the interactive prover. Examples are
Sledgehammer~\cite{Sledgehammer2007} for Isabelle,
HolyHammer~\cite{HolyHammer2014} for HOL4, and a hammer for
Coq~\cite{czajka2018hammer}.
A general overview is given in the survey paper by Blanchette
et al.~\cite{HammersSurvey2016}.
Some of these use learning to predict which premises are needed to be sent to
the ATP, in order not to overwhelm the prover.
In contrast to Lassie, the main focus of such hammers is not to make the
proofs more accessible but to solve simple proof obligations using a
push-button method.
As Lassie is open to adding custom decision procedures we think that
integrating a hammer with Lassie could provide for even richer and easier to
define 
tactic languages by automating simple proofs.

\paragraph{Learning-based}
While hammers try to automate the proof with the help automated theorem provers,
other systems use statistical methods to recommend tactics to the end user to
finish a proof.
DeepHOL~\cite{DeepHOL2019} learns a neural network that, given a proof goal,
predicts a potential next tactic in HOL Light.
GamePad~\cite{GamePad2019} and the work by Yang et al.~\cite{yang2019learning}
similarly use machine learning to predict tactics for Coq.
TacticToe~\cite{TacticToe2017} uses A* search, guided by previous tactic-level
proofs, to predict tactics in HOL4.

\sloppy

\paragraph{Programming Language-based}
Languages like Eisbach~\cite{Eisbach}, Ltac~\cite{Ltac},
 Ltac2~\cite{Ltac2} and Mtac2~\cite{Mtac2} use rigorous programming language
foundations to give more
control to expert users when writing tactics.
Eisbach and Ltac are tactic languages similar to the one of HOL4.
Mtac2 formalizes ``Coq in Coq'' allowing to define tactics as Coq programs,
whereas Ltac2 is a strongly typed language for writing Coq tactics.
The tactic language of the Lean theorem prover~\cite{de2015lean} additionally
implements equational reasoning on top of its tactics, which allows for more
textbook-like proofs.
Recently, the Lean theorem prover has also been extended with a hygienic
macro system~\cite{ullrich2020beyond}.
A core contribution of their work is excluding unintentional capturing in
tactic programming, thus making tactic programming more robust.
In \Tool{} we did not experience any hygiene issues as the definition by example
relies on the semantic parser to do the generalization and as such keeps variable
levels separate.
Using any of the languages above requires all the desired generality to be
stated explicit in the tactic definition, usually in the form of function
definitions.
In contrast, \Tool{}'s definition by example makes it easier to define
new tactics and generalizes automatically.


\fussy

\paragraph{Natural Language Interfaces}
Several systems provide an interface to a theorem prover that is as close
as possible to natural language.
Languages like Isar~\cite{Isar1999}, Mizar~\cite{Mizar2015}, and the work by
Corbineau~\cite{corbineau2007} follow a similar approach as Lassie by having an
extended parser.
Their supported naturalized proof descriptions are fixed to the authors style of
declarative proofs and extending or changing these would required editing the
tool code.
In contrast, \Tool{} is extensible enough to support different tactic languages
that can coexist without interferring if not loaded simultaneously.

The Naproche system~\cite{NaprocheSAD} provides a controlled natural language,
which maps natural language utterances into first-order logic proof obligations,
to be checked by an (automated) theorem prover
(e.g. E Prover~\cite{EProver2013}).
The extensions to Alfa by Hallgren et al.~\cite{Alfa} also use natural language processing
technology to extend the Alfa proof editor with a more natural language.
The book by Ganesalingam~\cite{Ganesalingam2013} gives a comprehensive explanation of the relation
between natural language and mathematics.
Similarly, Ranta et al.~\cite{Ranta11} provide more sophisticated linguistic
techniques to translate between natural language and predicate logic.
An orthogonal approach to the above is presented in the work by Coscoy
et al.~\cite{CoqExplain}.
Instead of translating from natural language to tactics, they provide a
translation from Coq proof terms to natural language.
The main goal of these systems is to provide an interface that supports as
much natural language as possible.
A major limitation, however, is that their grammars are fixed, i.e. only the
naturalized tactics implemented by the authors is available.
Our work does not strive to be a full natural language interface, and in turn
provides an extensible grammar, which adapts to different users and proofs.


%% file: sections/conclusion.tex
\section{Conclusion}\label{sec:conc}

We have presented the \Tool{} tactic language framework for the HOL4 theorem
prover.
Using a semantic parser with an extensible grammar, \Tool{} learns
individualized tactics from user-provided examples.
Our example case studies show that these learned tactics can be easily
reused across different proofs and can ease both the writing and reading
of HOL4 proofs by providing a more intuitive, personalized interface to HOL4's tactics.

